\documentclass[a4paper,twocolumn,10pt,accepted=2020-02-17]{quantumarticle}
\pdfoutput=1
\usepackage[utf8]{inputenc}
\usepackage[english]{babel}
\usepackage[T1]{fontenc}
\usepackage{amsmath}

\usepackage{tikz}
\usepackage{lipsum}

\usepackage{epsfig}
\usepackage{amsfonts}
\usepackage{amssymb}
\usepackage{amsthm}
\usepackage{mathrsfs}
\usepackage{bbold}
 \usepackage{color}
\usepackage{array,float}
\usepackage{longtable}

\usepackage{subfigure,diagbox}
\usepackage[numbers,sort&compress]{natbib}

\usepackage{optidef} 

\usepackage{algorithm} 
\usepackage{algpseudocode}

\usepackage{ftnxtra}

\newcolumntype{L}{>{\centering\arraybackslash}m{3cm}}

\newtheorem{theorem}{Theorem} 
\newtheorem{lemma}[theorem]{Lemma} 
\newtheorem{corollary}[theorem]{Corollary}

\newtheorem{definition}[theorem]{Definition}

\usepackage{xspace}
\newcommand{\methodname}{\textsc{SparSto}\xspace} 
\newcommand{\methodCOS}{\textsc{R1oTrott}\xspace}

\newcommand{\ket}[1]{\ensuremath{\vert#1\rangle}}

\def\>{\rangle} 
\def\<{\langle}

\newcommand{\expectation}[1]{\mathbb E \left({#1}\right)}  
\newcommand{\expectationinline}[1]{\mathbb E ({#1})}

\def\floor#1{{\left \lfloor {#1} \right \rfloor  }}
\def\ceil#1{{\left  \lceil  {#1} \right \rceil   }}

\begin{document}

\title{Compilation by stochastic Hamiltonian sparsification}

\author{Yingkai Ouyang}
\affiliation{%
 Department of Physics and Astronomy, University of Sheffield, Sheffield, UK
}%
\orcid{0000-0003-1115-0074}
\email{y.ouyang@sheffield.ac.uk}

\author{David R. White}
\email{d.r.white@sheffield.ac.uk}   
\affiliation{%
 Department of Physics and Astronomy, University of Sheffield, Sheffield, UK
}%
\orcid{0000-0001-6317-4348}

\author{Earl T. Campbell}
\affiliation{%
 Department of Physics and Astronomy, University of Sheffield, Sheffield, UK
}%
\affiliation{%
 Riverlane, Cambridge, UK
}%
\email{earltcampbell@gmail.com}

\maketitle

\begin{abstract}
 Simulation of quantum chemistry is expected to be a principal application of quantum computing. In quantum simulation, a complicated Hamiltonian describing the dynamics of a quantum system is decomposed into its constituent terms, where the effect of each term during time-evolution is individually computed. For many physical systems, the Hamiltonian has a large number of terms, constraining the scalability of established simulation methods. To address this limitation we introduce a new scheme that approximates the actual Hamiltonian with a sparser Hamiltonian containing fewer terms. By stochastically sparsifying  weaker Hamiltonian terms, we benefit from a quadratic suppression of errors relative to deterministic approaches. Relying on optimality conditions from convex optimisation theory, we derive an appropriate probability distribution for the weaker Hamiltonian terms, and compare its error bounds with other probability ansatzes for some electronic structure Hamiltonians. Tuning the sparsity of our approximate Hamiltonians allows our scheme to interpolate between two recent random compilers: qDRIFT and randomized first order Trotter.  Our scheme is thus an algorithm that combines the strengths of randomised Trotterisation with the efficiency of qDRIFT, and for intermediate gate budgets, outperforms both of these prior methods.
\end{abstract}

\section{Introduction}
Chemistry simulation is expected to be a principal application of quantum computing, revealing the properties of chemical bonds and interactions by simulating classically intractable systems, with applications in pharmaceuticals, material science, and industrial chemical manufacture \cite{AspuruGuzik2005,mccardle2018}. For example, efficient simulation of the chemical cluster FeMoCo \cite{Beinert1997,Reiher2017,babbush-1902.02134} may allow more efficient nitrogen fixation, improving the manufacture of fertilizers.
The difficulty of directly solving the Schr\"{o}dinger equation of chemically interesting problems renders even moderately complex systems classically intractable. Using quantum mechanics to simulate quantum systems \cite{Lloyd1996} will provide unprecedented detail in the solution to chemical problems, enabling us to better understand the dynamics of highly entangled systems, predicting their properties and chemical reactions \cite{AspuruGuzik2005,PhysRevX.8.031022-trapped-ion}. For example, simulation can be used in phase estimation to extract the eigenspectrum of a Hamiltonian \cite{Whitfield2011,wiebe-1907.10070,PhysRevX.6.031007-scalable-qsim-of-molecular-energies}, and by sufficiently understanding its energy spectra we can accurately predict chemical reaction rates \cite{PhysRevLett.122.090504-qsim-h2o-transitions}.

Given a Hamiltonian $H$ expressed as the sum of multi-qubit Pauli matrices, solution of the Schr\"{o}dinger equation requires the calculation of its exponentiation.  Quantum simulation techniques are distinguished by the way they map 
the chemical Hamiltonian to an effective Hamiltonian on qubits \cite{babbush-1909.00028,Setia2018-BrK-simulation,PhysRevLett.120.110501-lineardepth,PhysRevX.8.011044-lowdepth,PhysRevX.8.041015-Tcomplexity}, and subsequent mapping of the exponentiation of this effective Hamiltonian to a computation. One well-established method is to apply the Trotter-Suzuki formula \cite{Suz76,Suzuki1990,Suzuki1991} to reduce this larger exponentiation to a product of Pauli exponentials. 
Each Pauli exponential can then be interpreted as a single quantum gate. 
To date, many variants of Trotter-Suzuki methods have been studied in the context of quantum simulation of chemical systems \cite{Babbush2015-PhysRevA.91.022311,babbush-1902.10673,Somma2016-TrotterSuzukiJMP}.

Trotterisation is an attractive approach to quantum simulation because of its simplicity. However, a significant limitation of Trotterisation is that the number of quantum gates required scales linearly with the number of terms in the Hamiltonian, which may grow very large \cite{Whitfield2011,Wecker2014-PhysRevA.90.022305}. 
Campbell \cite{campbell2018random} observed the problematic scaling of the Trotter-Suzuki decomposition and introduces qDRIFT, a stochastic approach that samples terms from the Hamiltonian, trading accuracy for computational cost; importantly, qDRIFT scales independently of the number of terms. We build on Campbell's approach, introducing an algorithm that combines the strengths of standard Trotterisation with the efficiency of qDRIFT.

Our approach approximates the target Hamiltonian via sparsification, yielding a Hamiltonian with fewer terms; the Trotterised computation then requires far fewer gates per Trotter step. Although sparsification introduces a new approximation error, the reduction in gate count means we can apply more Trotter steps within a fixed budget. Our key insight is that sparsification can be performed stochastically instead of deterministically, and that this leads to improved performance. By defining our approximate Hamiltonian to be a random variable with an expectation value equal to the actual Hamiltonian, we benefit from a quadratic suppression of errors relative to deterministic methods; this behaviour is also seen in other stochastic compilers \cite{campbell2017mixing,campbell2018random,childs2018faster,HastingsML}.  We refer to our combination of first order Trotterisation and stochastically sparsified Hamiltonians as \methodname.

In many systems, including electronic structure Hamiltonians, we empirically observe power-law distributions of term strengths, which is promising for sparsification since weaker terms are clear candidates for truncation.  In a stochastic compiler, it is natural to relate the probability of a term being truncated from the Hamiltonian with the magnitude of its strength.  One of our main technical results is a rigorous upper bound on the error of \methodname for an arbitrary probability distribution where the terms are sampled independently. 

To obtain the best possible upper-bound we need to select the best probability distribution. In our analysis, we place the most important terms, with the largest strengths, inside an \emph{active} set so that they always appear in the sparsified Hamiltonian.  Random sparsification is instead applied only to a tail of weaker Hamiltonian terms that we label the \emph{inactive} set. We rely on optimality conditions in convex optimisation theory in order to derive a probability distribution over the inactive set, which we call the ``linear ansatz''.  We numerically optimise and analyse the performance of our error bounds for some electronic structure Hamiltonians.  For low gate budgets \methodname behaves similarly to qDRIFT, and for larger gate budgets it exactly reproduces randomized first order Trotter; as such it interpolates between these approaches.  However, for intermediate gate budgets, which coincide with parameter regimes of practical interest, our new sparsification method outperforms both methods by around an order of magnitude.  We emphasize numerical optimisation is always performed at the level of upper bounds and not by considering empirical performance of small simulatable systems.  Though \methodname uses first order Trotter, sparsification of Hamiltonians could also be naturally combined with higher order Trotter schemes to yield higher order randomized compilers.

When comparing with other Trotter methods we can just count the number of gates of the form $e^{-i s H_j}$ that are unitaries generated by easily accessible Hamiltonians $H_j$. However, to fairly compare against post-Trotter methods \cite{ChildsBerry2012,Low2017,BERRY2017,babbush-1902.02134,Low2019-qubitization}, we would also need to assess the resource cost of extra ancilla and select and prepare gadgets that are not built using $e^{-i s H_j}$.  This makes direct comparison with post-Trotter methods an involved task and sensitive to the cost model.  However, post-Trotter methods have better asymptotic scaling and for problems of interest they have often been found to have a considerable advantage over Trotterisation methods, and it is likely that this advantage persists over \methodname.  However, the value of this work is two-fold.  Firstly, we improve the best known Trotter methods.  Secondly, we have demonstrated the value of randomly sparsifying Hamiltonians as a technique that might later be incorporated into post-Trotter methods.  Indeed, Berry's perspective article \cite{berry2019random} also highlighted that potential application of randomization to post-Trotter protocols is a promising future research direction.

\section{Trotterisation}

Consider a time-independent Hamiltonian that admits a decomposition $	H = \sum_{j=1}^L h_j P_j.$
While $P_j$ can often be considered to be multi-qubit Pauli matrices, we make no such assumption, and for us 
$P_j$ are matrices with singular value at most 1. Without loss of generality, the corresponding coefficients $h_j$ can then be positive.  Solving the Schr\"odinger equation $ \ket{\psi(t)} = e^{-iHt} \ket{\psi(0)}$ allows us to model the continuous evolution of the state $\ket{\psi(t)}$ over time.
Over a short time period $s$,
the first order Trotter-Suzuki decomposition approximates the exponential $e^{-iHt}$ with a product of exponentials given by \cite{Suz76}
\begin{equation}
    \prod_{j=1}^L e^{-i s H_j} = e^{-i s H} + \mathcal{O}(s^2)
    \label{eq:trotter},
\end{equation} 
where $H_j = h_j P_j$.
We call this the ``vanilla Trotterisation'' scheme, as there are many variants to this approach \cite{Suzuki1990,Suzuki1991,childs2018toward,childs2018faster},
including ones that propose to deterministically coalesce Hamiltonian terms \cite{coalescing-PhysRevA.90.022305-1312.1695,coalescing-QIC-1406.4920}.
We assume that each $e^{-i s H_j}$, which we call a gate, can be efficiently implemented on a target quantum computer.
To simulate $e^{-itH}$, we approximate
$e^{-isH}$ repeatedly $r$ times, such that $t=rs$.
The number of gates $G$ required by a simulation is thus the principal measure of its computational cost. 
Since each Trotterisation of $e^{-isH}$ involves $L$ gates, 
vanilla Trotterisation requires $G=rL$ gates, and can become potentially computationally expensive when $L$ is large. 
In particular, it is known that an effective Hamiltonian for the electronic structure problem for a system with $N$ modes typically has $L= \mathcal{O}(N^4)$ terms \cite{AspuruGuzik2005}. 



The vanilla Trotterisation scheme approximates $e^{-itH}$ with a simulation error that is at most
$    \epsilon_{\rm van} 
\lesssim   \lambda^2 t^2 / 2r
= L \lambda^2 t^2 / 2G $ \cite{Suz76},
where $\lambda = \|{\bf h}\|_1$, ${\bf h}
=  (h_1,\dots, h_L)$.
In contrast to Trotterisation, the qDRIFT method introduced by Campbell has a computational cost that is independent of $L$; its gate count is $\mathcal{O} (   \lambda^2 t^2 / \epsilon)$. qDRIFT simulates an ideal unitary process by a Markovian evolution, sampling a sequence of Pauli gates from a predetermined distribution; each exponentiation in the resulting circuit is given the same weight $\tau$ such that the distribution alone determines the outcome of the calculation. The probability $p_j$ of choosing a given $e^{i\tau H_j}$ as the next gate in a computation is weighted by the corresponding interaction strength: $p_j = h_j / \lambda$, ensuring that the stochastic process of repeated sampling drifts stochastically towards the target unitary. The number of gates is set at a fixed computational budget $G$ representing the number of primitive gates, and gives approximation error $\epsilon \lesssim 4\lambda^2t^2 /G$ \footnote{In \cite{campbell2018random}, the simulation error used is the diamond distance, which differs from the diamond norm of the difference of channels by a factor of 2. This explains why the bound in \cite{campbell2018random} is approximately at most $2\lambda^2t^2 /G$ instead of $4\lambda^2t^2 /G$}.

Whilst Trotter-Suzuki decompositions have worse computational complexity than qDRIFT in the number of terms, their computational cost scales better with respect to $t$ and $\epsilon$. To exploit this trade-off, we introduce a new approach \methodname, which interpolates between qDRIFT and the higher-order Trotter-Suzuki decompositions whilst also building and improving on the analysis of randomised simulation in Ref.~\cite{childs2018faster}.

Like qDRIFT, \methodname approximates a unitary evolution with a probabilistic ensemble of unitary evolutions instead of direct Trotterisation --- although higher-order Trotterisation can subsequently be applied. This stochastic scheme is in the spirit of related work \cite{childs2018faster,childs2018toward,berry2019time} and its merits lie in the ability to use mixtures of unitary operators to approximate a unitary operator \cite{campbell2017mixing,HastingsML}; intuitively stochastic methods avoid systematic noise.

\section{\methodname Analysis}

\methodname uses a random Hamiltonian $\hat{H}$ and crucial to our analysis is that the expectation value is equal to the system Hamiltonian $\mathbb{E}(\hat{H})=H$. To reduce gate counts, we would like $\hat{H}$ to have far fewer terms than $H$ and to be a good approximation, or at least to do this with high probability. Rather than considering arbitrary probability distributions we consider the term-wise independent sampling where $\hat{H}$ contains the term $h_j P_j / p_j$ with probability $p_j$ and with probability $1-p_j$ this term is dropped.  This ensures $\mathbb{E}(\hat{H})=H$ and that the expected number of terms is $\mu=\sum_{j=1}^L p_j$.  

Next, we review Lindblad's formalism of unitary maps \cite{Lin76}, where such maps are generated by exponentiating Liouville operators. We let $\mathcal L_j = h_j \mathcal P_j$ so that $\mathcal P_j$ is a Liouville operator that maps $\rho$ to $ -i ( P_j \rho - \rho P_j)$.
Clearly then, $\mathcal L_j(\rho) =  -i ( H_j \rho - \rho H_j)$.
For any positive number $s$, Liouville operators generate unitary evolutions in the sense that
$e^{s \mathcal L_j} (\rho)  = e^{-i H_j s } \rho e^{i H_j s } .$
The ideal evolution operator can be written in terms of the Liouville operator $\mathcal L = \sum_{j=1}^L	\mathcal L_j$, because $
e^{s \mathcal L} (\rho)  = e^{-i H s } \rho e^{i H s } .$
Using a vanilla Trotterisation analogous to that given in \eqref{eq:trotter}, a first order approximation of $e^{s \mathcal L}$ is given by $\mathcal T_{s,\to} = \prod_{j=1}^L e^{s \mathcal L_j}$.
Given no \textit{a priori} reason to simulate $\mathcal L_j$ in any particular order, it is natural to consider permutations of this operator sequence. A second order approximation of the Taylor expansion can be made by mixing the above Trotterisation $\mathcal T_{s,\to}$ with $\mathcal T_{s,\gets} = \prod_{j=L}^1 e^{s \mathcal L_j} $, where the arrows denote the ordering over the term index $j$. The uniformly mixed operation $\frac{1}{2} (\mathcal T_{s,\to}  + \mathcal T_{s,\gets})$, which is a randomised first order Trotterisation scheme that we denote as \methodCOS, approximates $e^{s \mathcal L}$ with error that is third order in $s L$ \cite[Theorem 1]{childs2018faster}. More generally, Trotterisation can be further improved by completely randomising the order in which the gates are performed \cite{childs2018faster,Wecker2014-PhysRevA.90.022305}.
   To approximate $e^{t\mathcal L}$ for a fixed time $t$,
we approximate $e^{s\mathcal L}$ for a total of $r$ times, where $s = t/r$. By taking the number of repeats $r$ to be large, the simulation time $s$ can become small, and \eqref{eq:trotter} holds to a good approximation. 

In \methodname, we stochastically sparsify the Hamiltonian using some term-wise independent probability distribution and then apply one step of randomized first order Trotter.  For each of the $r$ Trotter steps a fresh stochastic Hamiltonian is sampled. The random Hamiltonian $\hat{H}$ induces a random Liouville operator $\hat{\mathcal{L}}$ in the natural way $\hat{\mathcal{L}}(\rho) =  i ( \hat{H} \rho - \rho \hat{H})$. Similarly, we have random terms $\hat{\mathcal{L}}_j$ such that
\begin{equation}
\hat{\mathcal{L}}_j  = \begin{cases}
    \mathcal{L}_j / p_j  & \mbox{ with probability } p_j \\
    0  & \mbox{ with probability } 1- p_j \\
\end{cases}
\end{equation}
Here, $\hat{\mathcal L_j}$ approximates the ideal Liouville operator $\mathcal L_j$ in the sense that 
$\mathbb{E}( \hat {\mathcal L_j} )= {\mathcal L_j}$.   Given a sampled $\hat{H}$ or $\hat{\mathcal{L}}$, we also randomize the order of the gates in each Trotter step and so introduce the randomized operators of forward and reverse Trotter steps
\begin{align}
    \hat{\mathcal T}_{s,\to} & = \prod_{j=1}^L e^{s \hat{\mathcal L}_j} \\
    \hat{\mathcal T}_{s,\gets}&  = \prod_{j=L}^1 e^{s \hat{\mathcal L}_j} .
\end{align}
A single step of \methodname is then described by
\begin{align}
\hat{\mathcal E}_s = \frac{1}{2}\left( \hat{\mathcal T}_{s,\to} + \hat{\mathcal T}_{s,\gets} \right)
\end{align}
which has $\mu$ gates on average.
To approximate $e^{t \mathcal L}$,
\methodname simulates $\hat{\mathcal E}_s$ independently and sequentially $r$ times.
By fixing the expected number of gates $G$ of \methodname to be constant,
we require in the first $\floor{G / \mu}$ repeats to have $s = \mu t / G$
and in the final repeat to have $s = t - \floor{G / \mu}$. The total number of repeats is then $r = \ceil{G / \mu}$.
We do not consider completely randomising the gate orders in $\hat {\mathcal E}_s$ because it renders our subsequent analysis overly complicated.

We quantify the maximum error of \methodname using the diamond norm \cite{kitaev2002classical},
which when evaluated on the difference between quantum channels, quantifies their distinguishability.
For us, we quantify the distinguishability between the average of $r$ repeats of $\hat{\mathcal E}_s$ and the ideal channel $e^{t \mathcal L}$ with the error
\begin{equation}
    \| \expectationinline{ \hat{\mathcal{E}}_s^r} - e^{t \mathcal L} \|_ \diamond, \label{eq:error-defi}
\end{equation}
where $\|\cdot \|_\diamond$ denotes the diamond norm.

Our main result is an analytic upper bound on $\|\hat{\mathcal E}_s -e^{s \mathcal L}\|_\diamond$ that we denote as $\epsilon$, which is expressed in terms of the 1-norm $\|\cdot\|_1$.
\begin{theorem} \label{thm:main}
Using \methodname with vector of probabilities  ${\bf p} = (p_1,\dots, p_L)$, vector of Hamiltonian coefficients ${\bf h} = (h_1,\dots, h_L)$, $L \ge 3$, and expected number of gates $G$,
the error of simulating $e^{t \mathcal L}$ with \methodname is at most $\epsilon$ where  
\begin{align}
   \epsilon =    \frac{2 t^2 \mu}{G} \| {\bf u} \|_1
   + \frac{4t^3 \mu^2}{3G^2} K
+ \mathcal O\left( \frac{t^4 \mu^3 }{G^3} \right),\notag
\end{align} 
with $K =\left( \|{\bf v}\|_1 + \lambda \| {\bf w}\|_1    + 4 \lambda^3/3  \right) $,
$\lambda = \|{\bf h}\|_1$
and $\mu = \| {\bf p} \|_1$.
Moreover, ${\bf u}, {\bf v}$ and ${\bf w}$ are vectors given by
\begin{align} 
{\bf u} = 
\left(
\left( \frac{1}{p_1} - 1 \right) h_1^2,
\dots,
 \left( \frac{1}{p_L} - 1 \right) h_L^2 \right),\notag\\
{\bf v} = 
\left(
\left( \frac{1}{p_1^2} - 1 \right) h_1^3,
\dots,
 \left( \frac{1}{p_L^2} - 1 \right) h_L^3 \right),\notag\\
{\bf w} = 
\left(
\left( \frac{3}{p_1} - 1 \right) h_1^2,
\dots,
 \left( \frac{3}{p_L} - 1 \right) h_L^2 \right).  \notag
\end{align}
\end{theorem}
A tighter bound with full details on the higher order terms in $\epsilon$ is supplied in 
Theorem~\ref{thm:main-complete} of Appendix \ref{app:order2,3-anaylsis}.

\begin{figure*}[htbp]
\centering
\includegraphics[width=\textwidth]{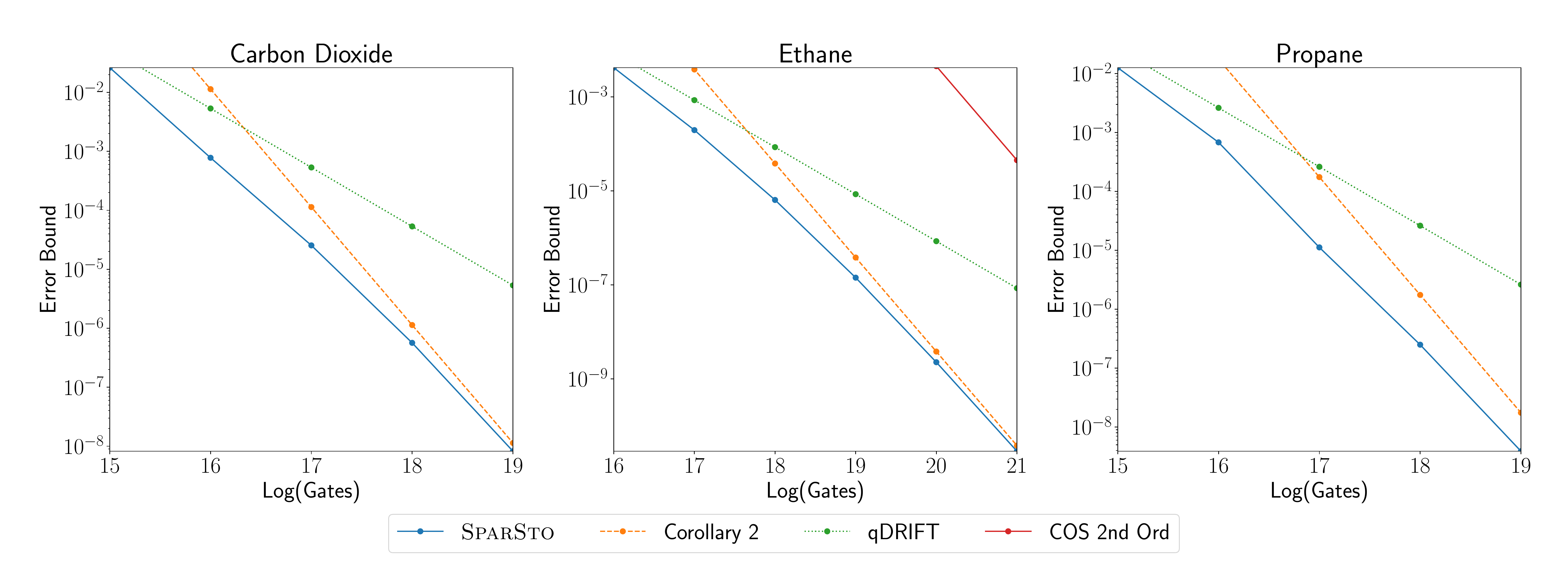}
\caption{Error Bounds: Rigorous upper bounds on the simulation errors of various molecules in the STO-3G basis set with $L \ge 100000$ are compared with rigorous bounds for Trotterisation and qDRIFT. Here $t=6000$. In an intermediate regime for expected the number of gates, \methodname requires fewer gates than both \methodCOS and qDRIFT for a fixed simulation error.
For propane and carbon dioxide, the second order Trotter error bounds (COS 2nd Order \cite[Theorem 2]{childs2018faster}) are too large to be seen on the plots.}
\label{fig:1}
\end{figure*}%

To bound the diamond distance between $\expectationinline{\hat{\mathcal E}_s^r}$ and $e^{t\mathcal L}$,
we bound the diamond distance between $\expectationinline{\hat{\mathcal E}_s}$ and $e^{s\mathcal L}$.
Using the triangle inequality on a telescoping sum, the unit diamond norm of all channels, and the independence of each random unitary $\hat{\mathcal E}_s$, 
we get the bound 
\begin{align}
\left\|\expectationinline{\hat{\mathcal E}_s}^r  - e^{t \mathcal L}  \right\|_\diamond \le r \left\|\expectationinline{\hat{\mathcal E}_s}   - e^{s \mathcal L}  \right\|_\diamond. 
\end{align}
To obtain an upper bound on $\|\expectationinline{\hat{\mathcal E}_s}   - e^{s \mathcal L}  \|_\diamond$, we perform a series expansion of the operators $\hat{\mathcal E}_s$ and $e^{s \mathcal L}$ with respect to the parameter $s$ to get
$\hat{\mathcal E}_s = \sum_{j \ge 0} \hat {\mathcal A}_j s^j$
and 
$e^{s \mathcal L} = \sum_{j \ge 0} \mathcal B_j s^j$. 
Using this notation,
we can see that $\hat {\mathcal A}_0 $ and $\mathcal B_0$ are both trivially the identity operator $\mathbb 1$, and 
$\expectationinline{\hat {\mathcal A}_1}$ and $\mathcal B_1$ are both equal to the Liouvillean $\mathcal L$.
To obtain the $\mathcal O(t^2\mu / G)$ and $\mathcal O(t^3\mu^2/G^2)$ terms in $\epsilon$, we evaluate upper bounds on 
$\| \expectationinline{\hat {\mathcal A}_2} - \mathcal B_2 \|_\diamond$
and $\| \expectationinline{\hat {\mathcal A}_3} - \mathcal B_3 \|_\diamond$ respectively. 
To do this, we rewrite $\hat {\mathcal A}_2$ 
as sums over products of $\hat{\mathcal L}_j^{k_j}$, where each sum comprises of terms of the form 
$\hat{\mathcal L}_j^{2}$, 
$\hat{\mathcal L}_j \hat{\mathcal L}_k$,
where $j$ and $k$ are distinct indices.
Having $j$ and $k$ distinct allows us to find that 
$\expectationinline{\hat{\mathcal L}_j^{2}} 
={{\mathcal L}_j^{2}/p_j}$
and
$\expectationinline{\hat{\mathcal L}_j \hat{\mathcal L}_k}
={{\mathcal L}_j {\mathcal L}_k}$.
Using a similar strategy for rewriting $\hat {\mathcal A}_3$,
we can evaluate its expectation explicitly.
Writing $\mathcal B_2$ and $\mathcal B_3$ in a similar form then allows us to compute the leading order terms in $\epsilon$. We supply the full details of this argument in Appendix \ref{app:order2,3-anaylsis}.

We upper bound the difference between tails of $\hat{\mathcal E}_s$ and $e^{s \mathcal L}$, which are $\mathcal O(t^4\mu^3 /G^3)$ terms, by essentially using the fundamental theorem of calculus to bound the tail of a power series from its derivatives.
From this, we evaluate upper bounds on the diamond norm of  $ \sum_{j\ge 4} s^j ( \expectationinline{\hat {\mathcal A}_j} - \mathcal B_j )$, and call this our tail bound.
To apply the fundamental theorem of calculus, we first take the fourth derivatives of $\hat{\mathcal E}_{s \theta}$ and $e^{s \theta \mathcal L}$ with respect to $\theta$, evaluate upper bounds on the norm of their difference over the unit interval for $\theta$. Second, we integrate this upper bound over an appropriate region, which gives a rescaling factor of $1/4!$.
Also, by obtaining polynomials in the diamond norms of $\mathcal L_j$ and subsequently using the inequality $\|\mathcal L_j\|_\diamond \le  2 h_j$,
along with the triangle inequality on the diamond norm of the difference between the ideal channel and the approximate channel,
we can obtain a closed form expression for the tail bounds which we show explicitly in Theorem~\ref{thm:main-complete} of Appendix \ref{app:order2,3-anaylsis}.

It is important to point out that the upper bound on the simulation error in Theorem~\ref{thm:main} depends very much on the choice of the probabilities $p_1,\dots ,p_L$.
Each $p_j$ signifies the probability that the Hamiltonian term $H_j$ contributes to the Trotterisation at each iteration. The smaller the value of $\mu = p_1 + \dots + p_L$, the sparser our Hamiltonian simulation is.
Intuitively, different choices on the values of the probabilities $p_j$ in Theorem~\ref{thm:main} affect the overall simulation error of $e^{t \mathcal L}$.
When all probabilities are equal to one, \methodname becomes identical to \methodCOS.
The simulation error, can thereby be obtained as the following corollary of Theorem~\ref{thm:main}.
\begin{corollary} \label{coro:epsilon-p=1}
When $p_1=\dots=p_L=1$, the simulation error is at most
\begin{align}
    \epsilon &= 
\frac{8t^3 L^2}{3G^2}  
\left(
\lambda \sum_{j=1}^L h_j^2
+
\frac{2\lambda^3}{3} 
\right)
+
\mathcal O\left(\frac{  t^4 L^3  }{ G^3}\right) . \notag
\end{align} 
\end{corollary}
While the bound that we have in Corollary \ref{coro:epsilon-p=1} is tighter than  \cite[Theorem 1]{childs2018faster}, a careful analysis of the third order terms in \cite[Theorem 1]{childs2018faster} yields the same expression as that given in Corollary \ref{coro:epsilon-p=1}. 

One might also observe that when all the probabilities in Theorem~\ref{thm:main} are set to $p_j=1$, we have $\|{\bf u}\|_1 = \|{\bf v}\|_1 = 0$ and 
$\|{\bf w}\|_1$ is minimized, which implies that $\epsilon / r$ which is roughly the simulation error per time segment $s$, is in fact minimized.
This leads one to wonder what advantage might be gained by setting the probabilities to be otherwise.
The solution to this conundrum lies in the penalty we pay in making such a choice.
In this scenario, each $\hat{\mathcal E}_s$ comprises of $\mu = L$ gates, and the overall error $\epsilon $ for simulating $e^{t \mathcal L}$ need not be optimized since $rs^{j} \sim t^j (\mu/G)^{j-1}$, which appears as coefficients in Theorem~\ref{thm:main}, is in fact maximized when $\mu = L$. The resultant algorithm simulates $e^{t \mathcal L}$, with an expected gate count of $G$ when $s=\mu t/G$ for all but the last repeat and $r = \ceil{G/\mu}$.

One can imagine \methodname to be analogous to another qDRIFT where $\mu=1$ so that the expected number of gates per time segment $s$ is equal to one.
The tradeoff in this scenario is that $\|{\bf u}\|_1$ and $\|{\bf v}\|_1$ are potentially very large because the probabilities become very small.
The key advantage of using Theorem~\ref{thm:main}  allows us to understand how $\epsilon$ interpolates between having all the probabilities to be either 1 or 0.
In what follows, we consider one family of probability distributions that we use together with Theorem~\ref{thm:main}.
For this example, we set $p_j=1$ whenever $h_j$ is above a set threshold.
Otherwise, $p_j < 1$.
We denote the active set $A$ as the set of indices $j$ for which $p_j =1$, and the inactive set $\bar A$ to be the set of indices for which $p_j<1$. 
We choose the values of $p_j$ according the following ansatz.
\begin{definition}[Linear ansatz]
For every $j\in A$, we set $p_j=1$.
For every $j \in \bar A$, we set $p_j = c h_j$.
We correspondingly have $\mu = |A| + c \sum_{j \in \bar A} h_j$.
\end{definition}
Clearly, $c$ has to be sufficiently small so that we indeed have $p_j<1$ for all indices $j$ in the inactive set.
By minimizing $\epsilon$ with respect to all possible values of $|A|$ and $\mu$ using our linear ansatz, we can determine which probabilities $p_j$ to use. These probabilities can be inputted into \methodname, which we describe in the pseudocode Algorithm 1.

\begin{figure}[htbp]
\centering
\includegraphics[width=0.5\textwidth]{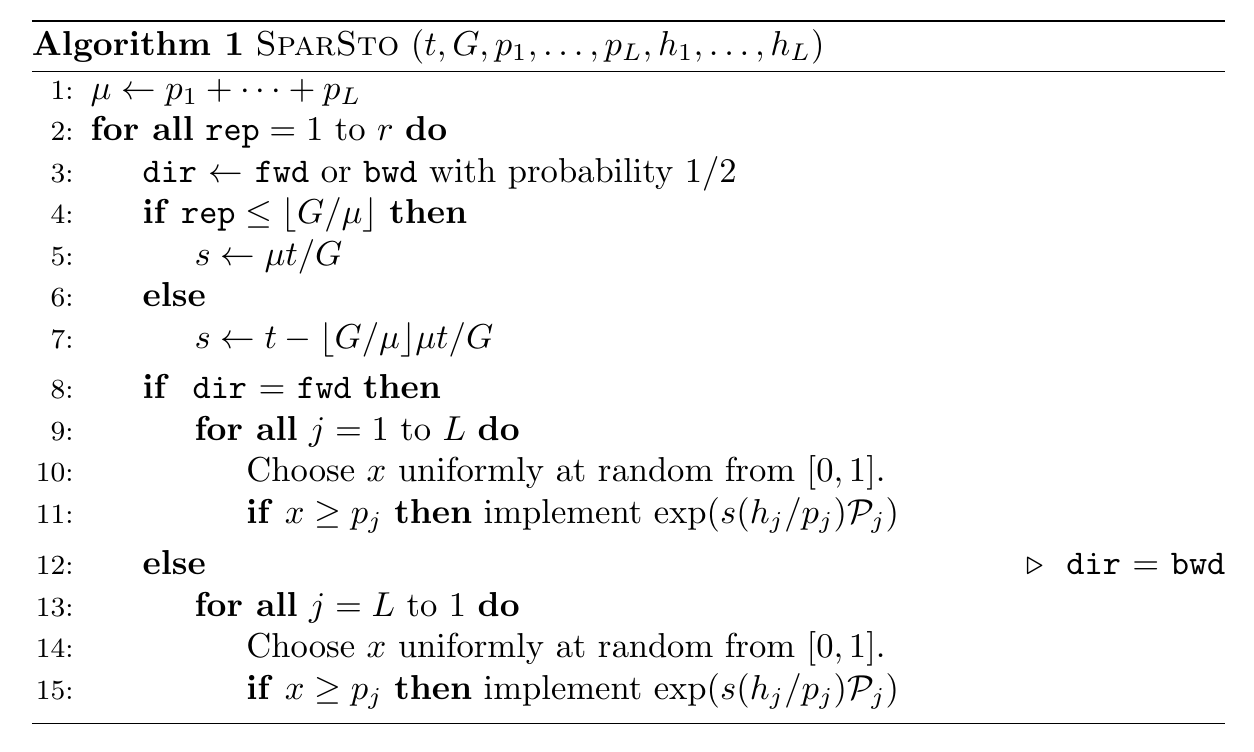} 
\end{figure}

We numerically study the performance of \methodname using models of molecules drawn from the OpenFermion library \cite{mcclean2017openfermion}, including carbon dioxide, ethane, and propane in the STO-3G basis set, and depict these results in Fig.~\ref{fig:1}. We evaluate the error bound for \methodname given by Theorem~\ref{thm:main-complete} in Appendix \ref{app:order2,3-anaylsis}. We compare the performance of \methodname with the Trotter bounds from \cite{childs2018faster}  (Theorem 2 in their paper, setting $k=1$), and by setting all probabilities $p_j=1$ we also plot Corollary \ref{coro:epsilon-p=1}.  Only the second order bounds from Childs {\em et al.} \cite[Theorem 2]{childs2018faster} are visible, in the upper-right of the second plot.

We perform a limited brute force numerical optimisation over all feasible values of $\mu$ and $|A|$ for our ansatzes; we examine $|A|/L$ over the interval $[0,1]$ with a step size of $0.1$, and consider the same values for $\mu' = (\mu-|A|)/(L-|A|)$ along with $1 \times 10^{-5}, 1\times10^{-4}, \text{ and } 1\times 10^{-3}$; we consider all pairwise combinations of these settings. 
Intuitively we expect that as the gate budget $G$ increases, we ought to interpolate between the qDRIFT regime \cite{campbell2018random} and the \methodCOS regime, 
and the size of the active set $|A|$ ought to go from 0 to $L$.
We observe from our numerical study that this indeed is the case.
In general, the optimal active set size increases with $G$, and the optimal value for $\mu'$ was usually small, and never more than $0.3$.

The linear ansatz outperforms the uniform ansatz. This is expected, as the uniform ansatz is na\"ive and the linear ansatz can be obtained as the optimal solution of the convex program which minimizes the leading order term in the total error for constant $\mu$ (see Appendix \ref{app:cvxopt}).
In each of the molecules, the number of Hamiltonian terms is over a hundred thousand, which is very large. When $t=6000$ and for a range of desired error values, there is a considerable advantage in using \methodname over both \methodCOS and qDRIFT. We observe similar results across other values of $t$ and smaller molecules.
 
\section{Discussion}
While vanilla Trotterisation can simulate any Hamiltonian with sufficiently many gates, the number of these gates can become very large. This leads to the need to reduce the gate count of quantum simulation while keeping the size of simulation error fixed. 
Here we present a new approach to chemistry simulation on a quantum machine, using the stochastic sparsification of a target Hamiltonian to derive a hybrid approach between canonical Trotterisation and qDRIFT. Our analysis provides an upper error bound for the scheme, and optimisation over the probabilities used in sparsification allows for reductions in the simulation error over parameter regimes of interest.

It would be instructive to consider how the ideas in our hybrid approach might extend to other variants of quantum simulation schemes, such as that of the so-called ``quantum signal processing'' (QSP) \cite{Low2017} techniques, linear combinations of unitaries \cite{BERRY2017}, the use of quantum walks \cite{ChildsBerry2012,babbush-1902.02134}, 
qubitisation \cite{Low2019-qubitization,babbush-1902.02134} and post-processing techniques \cite{wiebe-1907.13117}. 
There has also been recent interest in the quantum simulation of time dependent Hamiltonians \cite{wiebe-1805.00675,berry2019time}, 
and applications of quantum simulation in phase estimation \cite{campbell2018random,wiebe-1907.10070}, which may also prove amenable to stochastic sparsification. 
Given that random techniques can prove advantageous when applied to hybrid quantum-classical algorithms for numerical optimisation \cite{sweke2019stochastic}, our techniques might also offer some speedups in this area.
Furthermore, there might exist certain families of Hamiltonians where
the advantage of using our techniques over deterministic Trotterisation can be understood analytically, and we leave this as a subject for future work.

\textit{Acknowledgements}.- This work was supported by the EPSRC (grant no. EP/M024261/1), and has also received research funding from Huawei.  
We like to thank Yuan Su for a careful reading and comments on an earlier version of this manuscript.

\bibliography{qsim}{}
\bibliographystyle{plainnat}

\onecolumn
\newpage
\appendix

 \section{Upper bounds on the simulation error}
 \label{app:order2,3-anaylsis}
 In this section, we show that Theorem~\ref{thm:main} is a corollary of Theorem~\ref{thm:main-complete}, which we state in Section \ref{subsec:thm2-statement}.
 Before we can state Theorem~\ref{thm:main}, we define relevant notation in Section \ref{subsec:sums}.
 After that, we evaluate the leading order terms and tail terms of Theorem~\ref{thm:main-complete} in Section \ref{subsec:proof of leading order terms} and Section 
 \ref{subsec:proof of tail bounds} respectively.
\subsection{Sum over distinct indices}
\label{subsec:sums}

Given real vectors ${\bf a}=(a_1,\dots, a_n), {\bf b}= (b_1,\dots, b_n)$ and ${\bf c}= (c_1,\dots, c_n)$,
we define the sums over distinct indices to be
\begin{align}
\mathcal  S({\bf a}) 
&= \sum_{j=1}^n a_j\le \|{\bf a}\|_1,
\notag\\
\mathcal  S({\bf a},{\bf b}) 
&= 
\sum_{\substack{1 \le j,k\le n\\ j,k {\rm \ distinct}}} 
a_j b_k \le \|{\bf a}\|_1  \|{\bf b}\|_1, 
\notag\\
\mathcal  S({\bf a},{\bf b},{\bf c}) 
&= 
\sum_{\substack{1 \le j,k,l\le n\\ j,k,l {\rm \ distinct}}} 
a_j b_k c_l 
\le \|{\bf a}\|_1  \|{\bf b}\|_1  \|{\bf c}\|_1.
\label{eq:distinct-sums}
\end{align} 
To perform fast computation of the above sums, we can use the following lemma which
vectorises summations with distinct indices.
\begin{lemma}
\label{lem:combi-thirdmoment} 
Let $n$ be a positive integer. Let ${\bf a}=(a_1, \dots, a_n)$ and
 ${\bf b} = (b_1,\dots, b_n)$ be real column vectors.
Then 
\begin{align}
\mathcal S({\bf a}, {\bf b})
&=
 A_1 B_1  -  C_1 \label{relabel-trick1}
\\
\mathcal S({\bf a}, {\bf b}, {\bf b}) 
&= 
A_1 (B_1^2 - B_2)
- 2 C_1 B_1 + 2 C_2, \label{relabel-trick2a}
\\
\mathcal S({\bf a}, {\bf a}, {\bf a})
&= 
A_1^3 - 3 A_2 A_1 + 2 A_3, \label{relabel-trick2b}
\end{align}
where 
\begin{align}
A_j  &= \sum_{u=1}^n a_u^j,\\
B_j  &= \sum_{u=1}^n b_u^j,\\
C_j &= \sum_{u=1}^n a_u b_u^j.
\end{align}
\end{lemma}
Lemma \ref{lem:combi-thirdmoment} can be proved iteratively by careful consideration of summation indices.

\begin{proof}[Proof of Lemma \ref{lem:combi-thirdmoment}]
The result \eqref{relabel-trick1} is straightforward to show. 
To show \eqref{relabel-trick2a}, note that we can use \eqref{relabel-trick1} to write
\begin{align}
&\sum_{
\substack{
1 \le u,v,w \le n \\
u, v, w \ {\rm distinct}
}}
a_u b_v b_w  \notag\\
=&
\sum_{u=1}^n a_u 
\sum_{ \substack{ 1 \le v,w \le n \\ v ,w\ {\rm distinct} }} \!\! b_v b_w
-\!\!\!\!\!
\sum_{ \substack{ 1 \le u,w \le n \\ u ,w\ {\rm distinct} }} \!\!a_u b_u b_w
-\!\!\!\!\!
\sum_{ \substack{ 1 \le u,v \le n \\ u ,v\ {\rm distinct} }} \!\!a_u b_u b_v
\notag\\
=&
A_1 (B_1^2 - B_2)
- 2\sum_{u=1}^n a_u b_u B_1 + 2 \sum_{u=1}^n a_u b_u^2.\notag
\end{align}
We can specialize this to sum of distinct combinations of $a_u a_v a_w$ to get
\begin{align}
\sum_{
\substack{
1 \le u,v,w \le n \\
u, v, w\ {\rm distinct}
}}
a_u a_v a_w 
&= 
A_1 (A_1^2-A_2) - 2A_2 A_1 + 2A_3, \notag
\end{align}
which yields \eqref{relabel-trick2b}. 
\end{proof}

\subsection{Complete simulation error bound}
\label{subsec:thm2-statement}
The complete upper bound that we prove here is given by the following.
\begin{theorem} \label{thm:main-complete}
Using \methodname with vector of probabilities  ${\bf p} = (p_1,\dots, p_L)$, vector of Hamiltonian coefficients ${\bf h} = (h_1,\dots, h_L)$, $L \ge 3$, and expected number of gates $G$ where $G/(p_1+\dots +p_L)$ is an integer,
the error of simulating $e^{t \mathcal L}$ is at most $\epsilon$ where $\epsilon =  \epsilon_1 + \epsilon_2 + \epsilon_{3,1} + \epsilon_{3,2}  $ and
\begin{align}
    \epsilon_1 &=     \frac{2 t^2 \mu}{G} \mathcal S( {\bf u} ), \notag\\
    \epsilon_2 &=  \frac{4t^3 \mu^2}{3G^2}
\left(  \mathcal S( {\bf v} ) + \mathcal S( {\bf w}, {\bf h} ) \right)
+ \frac{16 t^3 \mu^2}{9G^2} 
\mathcal S({\bf h},{\bf h},{\bf h}),
\notag\\
\epsilon_{3,1} &= 
\frac{2t^4 \mu^3  \lambda^4  }{3G^3 }  ,
\notag\\
\epsilon_{3,2} &= 
\frac{2t^4 \mu^3 }{3 G^3  } 
(p_1 \dots p_L) \mathcal S({\bf q})^4
,\notag
\end{align} 
with $\mu = \sum_{j=1}^L p_j$.
Moreover, ${\bf u}, {\bf v}, {\bf w}$ and ${\bf q} $ are vectors given by
\begin{align} 
{\bf u} &= 
\left(
\left( \frac{1}{p_1} - 1 \right) h_1^2,
\dots,
 \left( \frac{1}{p_L} - 1 \right) h_L^2 \right),
 \notag\\
{\bf v} &= 
\left(
\left( \frac{1}{p_1^2} - 1 \right) h_1^3,
\dots,
 \left( \frac{1}{p_L^2} - 1 \right) h_L^3 \right),
 \notag\\
{\bf w} &= 
\left(
\left( \frac{3}{p_1} - 1 \right) h_1^2,
\dots,
 \left( \frac{3}{p_L} - 1 \right) h_L^2 \right),
 \notag\\
 {\bf q} &= \left( \frac{h_1}{p_1}, \dots ,  \frac{h_L}{p_L} \right)  .  \notag
\end{align}
\end{theorem}
We use $\mathcal S$ as defined in Section \ref{subsec:sums}.
We will see that by considering explicitly the commutation structure of the matrices $P_j$, we can obtain a tighter bound on $\epsilon_2$ in Theorem~\ref{thm:main-complete} by substituting $\mathcal S({\bf h},{\bf h},{\bf h})$ for $D_5$ in  \eqref{eq:commutator-bound}.

Before we proceed to prove Theorem~\ref{thm:main-complete}, we prove that Theorem~\ref{thm:main} is a straightforward consequence of Theorem~\ref{thm:main-complete}.
Note that 
\begin{proof}[Proof of Theorem~\ref{thm:main}]
By overcounting \eqref{eq:distinct-sums}, it is easy to see that 
\[
\mathcal S({\bf w} , {\bf h}) 
\le \|{\bf w} \|_1 \|{\bf h} \|_1
= \lambda \|{\bf w} \|_1 ,
\]
and
\[
\mathcal S({\bf h},{\bf h},{\bf h} )
\le \|{\bf h}\|_1  \|{\bf h}\|_1  \|{\bf h}\|_1
= \lambda^3.
\]
Moreover, since ${\bf u}$ and ${\bf v}$ are non-negative vectors, we have $\mathcal S({\bf u}) = \|{\bf u}\|_1$
and $\mathcal S({\bf v}) = \|{\bf v}\|_1$.
Furthermore, we have $\epsilon_{3,j} =\mathcal O(t^4 \mu^3/G^3)$.
This completes the proof.
\end{proof}

The proof of Theorem~\ref{thm:main-complete} then arises from the evaluation of (1) the leading order terms $\epsilon_{1}$ and $\epsilon_2$, and (2) the higher order terms $\epsilon_{3,1}$ and $\epsilon_{3,2}$. 
This will proceed in the next two subsections.
We emphasize that in what follows, because of the telescoping argument we mentioned in the main text, 
it suffices to only analyze $\|\hat{\mathcal E}_s- e^{s \mathcal L}\|_\diamond$, and the overall simulation error will just be $r$ times of this diamond norm.  

\subsection{The leading order terms $\epsilon_1$ and $\epsilon_2$ in Theorem~\ref{thm:main-complete}}
\label{subsec:proof of leading order terms}
Here, we show that the leading order terms in the simulation error are as given by $\epsilon_1$ and $\epsilon_2$.
First recall that we have the Taylor series expansions 
$\hat{\mathcal E}_s = \sum_{j \ge 0} \hat {\mathcal A}_j s^j$
and 
$e^{s \mathcal L} = \sum_{j \ge 0} \mathcal B_j s^j$.

Note that when $L\ge 3$, we have
\begin{align}
\hat{\mathcal T}_{s,\to}
&=
\prod_{j=1}^L 
\left(
{\bf 1} + s \hat{\mathcal L_j} + \frac{s^2}{2} \hat{\mathcal L_j} ^2  
+ \frac{s^3}{6} \hat{\mathcal L_j} ^3
+
\dots
\right) 
\notag\\
&=
{\bf 1} + s \sum_{j=1}^L \hat{\mathcal L_j} + \frac{s^2}{2} \sum_{j=1}^L \hat{\mathcal L_j} ^2  
+ \frac{2s^2}{2}  \sum_{1\le j<k\le L} \hat{\mathcal L_j} \hat{\mathcal L_k}  \notag\\
&\quad
+
\frac{s^3}{6}\sum_{j=1}^L \hat{\mathcal L}_j^3
+ 
s\frac{s^2}{2}\sum_{1\le j<k\le L}
\left(  \hat{\mathcal L}_j^2  \hat{\mathcal L_k} + \hat{\mathcal L_j}  \hat{\mathcal L}_k^2 \right)
\notag\\
&\quad +
s^3 \sum_{1\le j < k < l \le L} \hat{\mathcal L_j}  \hat{\mathcal L_k}  \hat{\mathcal L_l}  + \dots.\label{eq:trotter-forwards}
\end{align}
Similarly, 
\begin{align}
\hat{\mathcal T}_{s,\gets}
&=  
{\bf 1} + s \sum_{j=1}^L \hat{\mathcal L_j} + \frac{s^2}{2} \sum_{j=1}^L \hat{\mathcal L_j} ^2  
+ \frac{2s^2}{2}  \sum_{1\le k<j\le L} \hat{\mathcal L_j}
\hat{\mathcal L_k} \notag\\
&\quad
+
\frac{s^3}{6}\sum_{j=1}^L \hat{\mathcal L}_j^3
+ 
s\frac{s^2}{2}\sum_{1\le j<k\le L}
\left(  \hat{\mathcal L}_j^2  \hat{\mathcal L_k} + \hat{\mathcal L_j}  \hat{\mathcal L}_k^2 \right)
\notag\\
&\quad +
s^3 \sum_{1\le l < k < j \le L} \hat{\mathcal L_j}  \hat{\mathcal L_k}  \hat{\mathcal L_l}  + \dots.
\label{eq:trotter-backwards}
\end{align}
Since $\hat{\mathcal E}_s = \frac{1}{2}\left(
\hat{\mathcal T}_{s,\to}
+
\hat{\mathcal T}_{s,\gets} \right)$,
\eqref{eq:trotter-forwards} and  \eqref{eq:trotter-backwards} imply that
\begin{align}
 {2! \hat {\mathcal A}_2} &=  \sum_{j=1}^L \hat{\mathcal L_j} ^2   +  \sum_{j \neq k} \hat{\mathcal L_j} \hat{\mathcal L_k} , \\
{3! \hat {\mathcal A}_3}&=  
\sum_{j=1}^L \hat{\mathcal L}_j^3
+ 
\frac{6}{4}\sum_{j \neq k}
\left(  \hat{\mathcal L}_j^2  \hat{\mathcal L_k} + \hat{\mathcal L_j}  \hat{\mathcal L}_k^2 \right)  +
\frac{6}{2} 
\sum_{\substack{ 1\le j < k < l \le L\\ 1\le l < k < j \le L} }
 \hat{\mathcal L_j}  \hat{\mathcal L_k}  \hat{\mathcal L_l} .
\end{align}
Moreover, we know that 
\begin{align}
\mathcal L ^2 &= \sum_{j=1}^L \mathcal L_j^2  + \sum_{j \neq k } \mathcal L_j \mathcal L_k, \\
\mathcal L^3 &= \sum_{j=1}^L \mathcal L_j^3  + \sum_{j \neq k } 
\left( \mathcal L_j^2 \mathcal L_k + \mathcal L_j \mathcal L_k^2 + 
\mathcal L_j \mathcal L_k \mathcal L_j\right) 
\notag\\
&\quad + 
\sum_{\substack{ 1\le j < k < l \le L\\ 1\le l < k < j \le L} }
 {\mathcal L_j}   {\mathcal L_k}   {\mathcal L_l} 
 + 
\sum_{\substack{ 1\le k< j < l \le L\\ 1\le k < l < j \le L} }
 {\mathcal L_j}   {\mathcal L_k}   {\mathcal L_l} +  \sum_{\substack{ 1\le j < l < k \le L\\ 1\le l < j < k \le L} }
  {\mathcal L_j}   {\mathcal L_k}   {\mathcal L_l} .
\end{align}
Clearly $\mathcal L - \expectationinline{\hat { \mathcal L}} = 0 $.
Next note that 
\begin{align}
\mathcal L^2 -  \expectationinline{ {2! \hat {\mathcal A}_2}} =&
 \sum_{j=1}^L 
 \left( \mathcal L_j^2 - \expectationinline{ \hat{\mathcal L}_j^2 } \right)
   +  \sum_{j \neq k}
\left(   \mathcal L_j \mathcal L_k - \expectationinline{ \hat{\mathcal L}_j \hat{\mathcal L_k} } \right )
   \notag\\
=&
 \sum_{j=1}^L 
 \left( \mathcal L_j^2 - \expectationinline{ \hat{\mathcal L}_j ^2 } \right).
\end{align}
Now $ \expectationinline{\hat{\mathcal L}_j^2} = p_j \frac{\mathcal L_j^2}{p_j^2}$, which implies that for $0 < p_j \le 1$, we have
\begin{align}
\mathcal L^2 -  \expectationinline{ {2! \hat {\mathcal A}_2}} 
&= \sum_{j=1}^L 
 \left( 1-  \frac{1}{p_j }  \right)  \mathcal L_j ^2.
 \label{eq:diff-order2}
\end{align} 
Since $p_j \le 1$, we have
\begin{align}
\| \mathcal L^2 -  \expectationinline{ {2! \hat {\mathcal A}_2}} \|_\diamond 
\le
 \sum_{j=1}^L \left(  \frac{1}{p_j } - 1  \right) (4h_j^2).
\end{align}
Since $\mathcal B_2 = \mathcal L^2 / 2!$, we get the upper bound
\begin{align}
\|s^2 \mathbb E( \hat {\mathcal A}_2) - s^2 \mathcal B_2\|_\diamond 
&= \frac{s^2}{2!}
4\sum_{j=1}^L\left(  \frac{1}{p_j } - 1  \right) h_j^2\notag\\
&= 
2s^2
\sum_{j=1}^L\left(  \frac{1}{p_j } - 1  \right) h_j^2,
\end{align}
where $s = t \mu/G$.
 Multiplying this by $r=G/\mu$ gives us $\epsilon_1$.

To evaluate $\epsilon_2$, we proceed to write
\begin{align}
   \mathcal L^3 -  \expectationinline{ {3! \hat {\mathcal A}_3}} 
   &= 
   D_1 + D_2 + D_3 + D_4 + D_5,
\end{align}
where
\begin{align}
    D_1 &=  \sum_{j=1}^L \left( \mathcal L_j^3 -  \expectationinline{\hat{\mathcal L}_j^3}   \right),  \\
    D_2 &= \sum_{j \neq k }\left( \mathcal L_j^2 \mathcal L_k  - \expectationinline{ \hat{\mathcal L}_j^2 \hat{\mathcal L}_k} \right) ,\\
    D_3 &= \sum_{j \neq k } \left( \mathcal L_j \mathcal L_k^2 - \expectationinline{\hat{\mathcal L_j} \hat{\mathcal L}_k^2} \right),\\
    D_4 &= \sum_{j \neq k } \left( \mathcal L_j \mathcal L_k \mathcal L_j  - 
\frac{1}{2} \expectationinline{ \hat{\mathcal L}_j^2 \hat{\mathcal L_k} + \hat{\mathcal L_j} \hat{\mathcal L}_k^2} \right), \\
    D_5 &=\sum_{\substack{ 1\le k< j < l \le L\\ 1\le k < l < j \le L} }
 {\mathcal L_j}  {\mathcal L_k}  {\mathcal L_l}  
 + 
\sum_{\substack{ 1\le j < l < k \le L\\ 1\le l < j < k \le L} }
 {\mathcal L_j}  {\mathcal L_k}  {\mathcal L_l} -2 \sum_{\substack{ 1\le j < k < l \le L\\ 1\le l < k < j \le L} }
\expectationinline{ \hat{\mathcal L_j}  \hat{\mathcal L_k}  \hat{\mathcal L_l} } \label{eq:D5eq30}.  
\end{align}
We now proceed to simplify $D_j$ for $j=1,\dots, 5$.
Note that 
$\expectationinline{\hat{\mathcal L}_j^3} 
=p_j \frac{\mathcal L_j^3}{p_j^3}$.
This implies that
\begin{align}
    D_1 = \sum_{j=1}^L \left( 1-  \frac 1 {p_j^2}  \right) \mathcal L_j^3.
\end{align}
Next, multiplicativity of the expectation for independent random variables implies that
\begin{align}
    D_2 &= \sum_{j \neq k}  \left( 1- \frac 1 {p_j}  \right)
    \mathcal L_j^2 \mathcal L_k,
    \\
    D_3 &= \sum_{j \neq k}  \left(  1-\frac 1 {p_k}  \right)
    \mathcal L_j \mathcal L_k^2.
\end{align}
Now we can write
\begin{align}
        D_4 &= 
        \frac 1 2 \sum_{j \neq k }  \mathcal L_j \mathcal L_k \mathcal L_j  
        + \frac 1 2 \sum_{j \neq k }  \mathcal L_k \mathcal L_j \mathcal L_k - 
        \sum_{j \neq k } \left(
\frac{1}{2} 
\expectationinline{ \hat{\mathcal L}_j^2 \hat{\mathcal L_k} 
+ 
\hat{\mathcal L_j} \hat{\mathcal L}_k^2} \right) .\label{eq:A29}
\end{align}
Clearly, we have 
$\expectation{ \hat{\mathcal L}_j^2 \hat{\mathcal L_k} } = \mathcal L_j^2 \mathcal L_k/p_j$
and
$\expectation{ \hat{\mathcal L_j} \hat{\mathcal L}_k^2 } = \mathcal L_j \mathcal L_k^2/p_k$.
Next by swapping the roles of $j$ and $k$ in the summation, we get
\[ \sum_{j \neq k } \mathcal L_j \mathcal L_k^2/p_k
=
 \sum_{j \neq k }  \mathcal L_k \mathcal L_j^2/p_j.
\] 
By pairing the first term with the third term and the second term with the fourth term in \eqref{eq:A29}, this implies that  
\begin{align}
D_4 &=
\frac 1 2\sum_{j \neq k} 
\mathcal L_j \left(\mathcal L_k  \mathcal L_j - \mathcal L_j  \mathcal L_k /p_j  \right)
+
\frac 1 2 \sum_{j \neq k}
\left(
\mathcal L_j \mathcal L_k
-
\mathcal L_k \mathcal L_j / {p_k}
\right) \mathcal L_k ,
\end{align}
where we swap the roles of $j$ and $k$ in the second sum.
From the above, we can see that
\begin{align}
    \|D_1\|_\diamond &\le \sum_{j=1}^L 
    \left(  \frac 1 {p_j^2} -1 \right) 8 h_j^3,\\
    \|D_2\|_\diamond &\le \sum_{j \neq k}  \left(  \frac 1 {p_j} -1 \right) 8 h_j^2 h_k,\\
    \|D_3\|_\diamond &\le \sum_{j \neq k}  \left(  \frac 1 {p_k} -1 \right) 8 h_j h_k^2,\\
    \|D_4\|_\diamond &\le \sum_{j \neq k}  \left( 1 +  \frac 1 {p_j}  \right) 8 h_j^2 h_k.
\end{align}
From this, we can obtain the first two terms in $\epsilon_2$.
To see this, note that
\begin{align}
 \frac{1}{3!}\|D_1\|_\diamond 
 \le \frac{4}{3} \sum_{j=1}^L 
    \left(  \frac 1 {p_j^2} -1 \right)   h_j^3,\label{eq:D1-bound}
\end{align}
and
\begin{align}
   \frac{1}{3!}
   \sum_{i=2}^4
   \|D_i\|_\diamond   
   &\le
   \frac{4}{3} \sum_{j \neq k}  \left( 3 \frac 1 {p_j} -1  \right) h_j^2 h_k.\label{eq:D2,3,4-bound}
\end{align}
Multiplying the right sides of \eqref{eq:D1-bound} and \eqref{eq:D2,3,4-bound} by $r$ gives the first two terms in $\epsilon_2$.

We proceed to simplify $D_5$.
Note from \eqref{eq:D5eq30} that
\begin{align}
    D_5 = &
\sum_{\substack{ 1\le k< j < l \le L } }
 {\mathcal L_j}  {\mathcal L_k}  {\mathcal L_l}  + \sum_{\substack{   1\le k < l < j \le L} }
 {\mathcal L_j}  {\mathcal L_k}  {\mathcal L_l}  + 
\sum_{\substack{ 1\le j < l < k \le L } }
 {\mathcal L_j}  {\mathcal L_k}  {\mathcal L_l} \\ \nonumber
 & + \sum_{\substack{  1\le l < j < k \le L} }
 {\mathcal L_j}  {\mathcal L_k}  {\mathcal L_l} 
 -2 \sum_{\substack{ 1\le j < k < l \le L } }
  {\mathcal L_j}  {\mathcal L_k}  {\mathcal L_l}
 -2 \sum_{\substack{   1\le l < k < j \le L} }
 {\mathcal L_j}   {\mathcal L_k}   {\mathcal L_l}  .  
\end{align} 
Now by ordering all the indices in the same way we get
\begin{align}
    D_5 &=
\sum_{\substack{ 1\le j< k < l \le L } }
 {\mathcal L_k}  {\mathcal L_j}  {\mathcal L_l}  
+
\sum_{\substack{   1\le j < k < l \le L} }
 {\mathcal L_l}  {\mathcal L_j}  {\mathcal L_k}  
 + \sum_{\substack{ 1\le j < k < l \le L } }
 {\mathcal L_j}  {\mathcal L_l}  {\mathcal L_k}  
 \notag\\
& \quad + \sum_{\substack{  1\le j < k < l \le L} }
 {\mathcal L_k}  {\mathcal L_l}  {\mathcal L_j} 
 -2 \sum_{\substack{ 1\le j < k < l \le L } }
  \left(
  {\mathcal L_j}  {\mathcal L_k}  {\mathcal L_l}
  + 
   {\mathcal L_l}   {\mathcal L_k}   {\mathcal L_j} 
   \right).  \label{eq:A38}
\end{align} 
By pairing the first term with the fifth term, and the third term with the fifth term in \eqref{eq:A38}, we get
$\mathcal L_k \mathcal L_j \mathcal L_l - 
 \mathcal L_j \mathcal L_k \mathcal L_l
 = [\mathcal L_k, \mathcal L_j] \mathcal L_l$
and 
$\mathcal L_j \mathcal L_l \mathcal L_k - 
 \mathcal L_j \mathcal L_k \mathcal L_l
 = \mathcal L_j [\mathcal L_l, \mathcal L_k]$.
By pairing the second term with the sixth term, and the fourth term with the sixth term in \eqref{eq:A38}, we get
$\mathcal L_l \mathcal L_j \mathcal L_k - 
\mathcal L_l \mathcal L_k \mathcal L_j
 = \mathcal L_l[\mathcal L_j, \mathcal L_k]$
and 
$\mathcal L_k \mathcal L_l \mathcal L_j - 
 \mathcal L_l \mathcal L_k \mathcal L_j
 = [\mathcal L_k, \mathcal L_l]  \mathcal L_j$.
We can thus rewrite \eqref{eq:A38} as
\begin{align}
    D_5 &=
\sum_{\substack{ 1\le j< k < l \le L } }
\left( [\mathcal L_k, \mathcal L_j] \mathcal L_l
+ \mathcal L_j [\mathcal L_l, \mathcal L_k] + \mathcal L_l [\mathcal L_j , \mathcal L_k]
+[\mathcal L_k, \mathcal L_l] \mathcal L_j
\right)
\end{align} 
Collecting the terms in the above summation in terms of commutators again, we get
\begin{align}
    D_5 &=
\sum_{\substack{ 1\le j< k < l \le L } }
\left( 
[\mathcal L_l , [\mathcal L_j, \mathcal L_k] ]
+[[\mathcal L_k, \mathcal L_l], \mathcal L_j]
\right).
\label{eq:commutator-bound}
\end{align}
A trivial upper bound on the diamond norm of this is 
\begin{align}
    \| D_5\|_\diamond \le 8 \frac{8}{6} \mathcal S({\bf h},{\bf h},{\bf h}),
\end{align}
where the first factor of 8 arises from going from the diamond norm of the Liovillean $\mathcal L_j$ to the operator norm of $H_j$, and the numerator $8$ in the fraction arises from the total number of summations over non-decreasing indices, and 6 arises from the number of ways to permute the indices $j,k$ and $l$.
From the bounds we have on the diamond norms of $D_1,D_2,D_3,D_4$ and $D_5$, we obtain
\begin{align}
\frac{1}{3!}\left\| \mathcal L^3 -  \expectation{ {3! \hat {\mathcal A}_3}} \right\|_\diamond
   &\le 
   \frac{1}{6} \sum_{j=1}^5 \|D_j\|_\diamond
   = \epsilon_{2}/r.
\end{align}
Multiplying this by $r$ gives us the error in $\epsilon$ that is $\mathcal O(t^3\mu^2/G^2 )$.
We have thus completed bounding the leading order errors in Theorem~\ref{thm:main-complete}.

\subsection{Tail bounds in Theorem~\ref{thm:main} }
\label{subsec:proof of tail bounds}

Here, we explain how the tail bounds $\epsilon_{3,1}$ and $\epsilon_{3,2}$ in Theorem~\ref{thm:main-complete} arise.

To evaluate the higher order terms in $\epsilon$, we consider a convergent power series in $\theta$ given by
 $   \mathcal F_{\theta} = \sum_{k \ge 0 } f_k \theta^k.$
Here, $f_k$ is independent of $\theta$, 
and $\mathcal F_{\theta}$ and $f_k$ belong to a Banach algebra. Define $[\theta^j] \mathcal F_\theta = f_j$ as the $j$th coefficient in the power series expansion of $\mathcal F_\theta$. A useful technique to bound quantities in a Banach algebra relies on the fundamental theorem of calculus, and has been used for example in Ref~\cite{BrH13} and Ref~\cite{berry2019time}.
This for example can be used to obtain the well-known integral form of the remainder term of the Taylor series of the power series $\mathcal F_{s}$ where $s > 0$.
\begin{lemma}
\label{lem:calculus trick}
Let $\theta_0=1,s>0$ and $\mathcal F_s = \sum_{k \ge 0 } f_k s^k$.
For every positive integer $t$, we have 
\begin{align}
&\sum_{k \ge t } f_k s^k
= 
\int_{0}^{\theta_0} d\theta_1
\dots 
\int_{0}^{\theta_{t-1}} d\theta_t
\frac{d^t}{d\theta_t} \mathcal F_{s \theta_t}.\notag
\end{align}
\end{lemma}
\begin{proof}
The proof of this is well-known but we provide the complete details for completeness. 
We first note that 
$\frac{d^t}{d\theta_t} \mathcal F_{s \theta_t} 
= \sum_{k \ge t} f_k (s \theta_t )^{k-t} 
k_{\underline t}$,
where $k_{\underline t} = (k)\dots(k-t+1)$ denotes the falling factorial.

Applying the fundamental theorem of calculus on monomials in $\theta_t$, we have
\begin{align}
\int_0^{\theta_{t-1}} d\theta_t  \frac{d^t}{d\theta_t} \mathcal F_{s \theta_t} =&
\int_0^{\theta_{t-1}} d\theta_t  \sum_{k \ge t} f_k (s \theta_t )^{k-t} k_{\underline t}
\notag\\
=&
 \sum_{k \ge t} f_k s^{k-t} \theta_{t-1}^{k-t-1} k_{\underline {t-1}}.\notag
\end{align}
Applying this argument iteratively gives the result.  
\end{proof}
Now let us denote a single timeslice of the ideal channel and \methodname for time $s\theta$ as $U_\theta = e^{s \theta \mathcal L} $ and ${\hat {\mathcal E}_{s\theta}} = 
\frac{1}{2} (e^{s \theta \hat {\mathcal L}_1 } \dots e^{s \theta \hat {\mathcal L}_L}
+
e^{s \theta \hat {\mathcal L}_L } \dots e^{s \theta \hat {\mathcal L}_1})$ respectively.
We proceed to evaluate the fourth derivatives of a single timeslice of $U_\theta$ and ${\hat V}_\theta$, which are respectively given by 
\begin{align}
\frac{d^4}{d\theta^4} U_\theta & = (s \mathcal L)^4 U_\theta  , \\
\frac{d^4}{d\theta^4} {\hat {\mathcal E}_{s\theta}} &= \frac{s^4}{2}  \!\!\!\!
\sum_{\substack{n_1 + \dots + n_L = 4\\ n_1,\dots, n_L\in \mathbb N}}\!\!
\binom{4}{n_1,\dots,n_L}
\hat {\mathcal L}_1 ^{n_1}
e^{s \theta \hat {\mathcal L}_1 } \dots 
\hat {\mathcal L}_L ^{n_L}
e^{s \theta \hat {\mathcal L}_L}
\notag\\
&+
\frac{s^4}{2}   \!\!\!\!
\sum_{\substack{n_1 + \dots + n_L = 4\\ n_1,\dots, n_L\in \mathbb N}}\!\!
\binom{4}{n_1,\dots,n_L}
\hat {\mathcal L}_L ^{n_L}
e^{s \theta \hat {\mathcal L}_L } \dots 
\hat {\mathcal L}_1 ^{n_1}
e^{s \theta \hat {\mathcal L}_1},
\end{align}
where in the second equation we used the general Leibniz rule and the $\binom{4}{n_1,\dots,n_L} = 4!/(n_1! \dots n_L!)$ denotes the multinomial coefficient.
The diamond norm of the tail of $U_\theta$ is therefore at most
\begin{align}
    \frac{s^4}{4!}\| \mathcal L^4 U_\theta \|_\diamond
    &\le
    \frac{s^4}{4!}\| \mathcal L^4 \|_\diamond \| U_\theta \|_\diamond
    \notag\\
    &\le
    \frac{(2s)^4}{4!} \lambda^4  \| U_\theta \|_\diamond\notag\\
    &=
    \frac{(2s \lambda)^4}{4!} ,
\end{align}
where the last equality arises because $U_\theta$ is a quantum channel.

Now note that by the linearity of the derivative and expectation operator,
\begin{align}
    \frac{d^4}{d\theta^4} \mathbb E({\hat E}_{s\theta})
=& \frac{s^4}{2}   \!\!\!\!
\sum_{\substack{n_1 + \dots + n_L = 4\\ n_1,\dots, n_L\in \mathbb N}}\!\!
\binom{4}{n_1,\dots,n_L}
\mathbb E(
\hat {\mathcal L}_1 ^{n_1}
e^{s \theta \hat {\mathcal L}_1 } \dots 
\hat {\mathcal L}_L ^{n_L}
e^{s \theta \hat {\mathcal L}_L})
\notag\\
&
+\frac{s^4}{2}   \!\!\!\!
\sum_{\substack{n_1 + \dots + n_L = 4\\ n_1,\dots, n_L\in \mathbb N}}\!\!
\binom{4}{n_1,\dots,n_L}
\mathbb E(
\hat {\mathcal L}_L ^{n_L}
e^{s \theta \hat {\mathcal L}_L } \dots 
\hat {\mathcal L}_1 ^{n_1}
e^{s \theta \hat {\mathcal L}_1}).
\end{align}
By the independence of every stochastic Trotter step, the expectation is multiplicative so that
\begin{align}
    \frac{d^4}{d\theta^4} \mathbb E({\hat E}_{s\theta})
 = &
\frac{s^4}{2}   \!\!\!\!
\sum_{\substack{n_1 + \dots + n_L = 4\\ n_1,\dots, n_L\in \mathbb N}}\!\!
\binom{4}{n_1,\dots,n_L}
\mathbb E(\hat {\mathcal L}_1 ^{n_1}
 e^{s \theta \hat {\mathcal L}_1 }) \dots 
\mathbb E(\hat {\mathcal L}_L ^{n_L}
 e^{s \theta \hat {\mathcal L}_L})
\notag\\
&
+\frac{s^4}{2}   \!\!\!\!
\sum_{\substack{n_1 + \dots + n_L = 4\\ n_1,\dots, n_L\in \mathbb N}}\!\!
\binom{4}{n_1,\dots,n_L}
\mathbb E(\hat {\mathcal L}_L ^{n_L}
 e^{s \theta \hat {\mathcal L}_L }) \dots 
\mathbb E(\hat {\mathcal L}_1 ^{n_1}
 e^{s \theta \hat {\mathcal L}_1}).
\end{align}
Using the triangle inequality and the submultiplicativity of the diamond norm, the diamond norm of the tail of $\hat {\mathcal E}_{s \theta}$ is at most 
\begin{align}
&
\sum_{\substack{n_1 + \dots + n_L = 4\\ n_1,\dots, n_L\in \mathbb N}}\!\!\!\!\!\!\!\!
{\frac{s^4\binom{4}{n_1,\dots,n_L}}{4!}}
\|\mathbb E(\hat {\mathcal L}_1 ^{n_1}
e^{s \theta \hat {\mathcal L}_1 })\|_\diamond \dots 
\|\mathbb E(\hat {\mathcal L}_L ^{n_L}
e^{s \theta \hat {\mathcal L}_L})\|_\diamond. \label{eq:use-multinomial-theorem-here}
\end{align}
When $n_j \ge 1$, we appeal to the Taylor series expansion for the exponential function, to get
\begin{align}
    \expectationinline{
    \hat {\mathcal L}_j ^{n_j} 
    e^{s \theta \hat {\mathcal L}_j }
    }
    &=
    \expectation{
    \sum_{k \ge 0}
    \frac{(s\theta)^k}{k!}\hat {\mathcal L}_j^{k+n_j}
    }
    \notag\\
    &=
    \sum_{k \ge 0}
    \frac{(s\theta)^k}{k!}
    \expectationinline{
    \hat {\mathcal L}_j^{k+n_j}
    }\notag\\
    &=
    \sum_{k \ge 0}
    \frac{(s\theta)^k}{k!} 
    {\mathcal L}_j^{k+n_j} / p_j^{k+n_j-1} 
    \notag\\
    &=
    p_j({\mathcal L}_j/p_j ) ^{n_j} 
    e^{s \theta  {\mathcal L}_j / p_j }.\label{expectation:random-product}
\end{align}
It is clear that \eqref{expectation:random-product} also holds when $n_j=0$.
Using \eqref{expectation:random-product}, we get the upper bound 
\begin{align}
     \|\expectationinline{
    \hat {\mathcal L}_j ^{n_j} 
    e^{s \theta \hat {\mathcal L}_j }
    }\|_\diamond
    \le p_j(\|{\mathcal L}_j\|_\diamond /p_j ) ^{n_j} .
    \label{neat-diamond-norm-identity}
\end{align}
Using \eqref{neat-diamond-norm-identity} with the multinomial theorem on \eqref{eq:use-multinomial-theorem-here}, the diamond norm of the tail of $\hat {\mathcal E}_{s \theta}$ is at most 
\begin{align}
    \frac{s^4(p_1 \dots p_L)}{4!}
    \left( \|\mathcal L_1\|_\diamond/p_1 + \dots + \|\mathcal L_L\|_\diamond/ p_L \right)^4.
\end{align}
Applying the identity $\|\mathcal L_j\|_\diamond \le 2 h_j$, we find that the diamond norm of the tail of $\hat {\mathcal E}_{s \theta}$ is at most 
\begin{align}
       \frac{(2s)^4(p_1 \dots p_L)}{4!}
    \left( h_1/p_1 + \dots + h_L/ p_L \right)^4.
\end{align} 
By setting $\theta = 1$ and multiplying the results that we obtained from the tail bounds on a single timeslice $s$ by a factor of $r$, we thus find that the
expected contribution to the simulation error from the tail bounds on $r$ repeats of $e^{s \mathcal L}$ and $\hat{\mathcal E}_s$ is at most $\epsilon_{3,1}$ and $\epsilon_{3,2}$ respectively, where
\begin{align}
\epsilon_{3,1} &=  \frac{2^4 rs^4}{4!} \lambda^4\notag\\
\epsilon_{3,2} &=  \frac{2^4 rs^4}{4!} (p_1 \dots p_L)\left(h_1/p_1 + \dots + h_L/ p_L) \right),
\end{align}
from which the result follows.

 \section{Convex programming on leading order error terms}
 \label{app:cvxopt}
Here we minimise the leading order term in the simulation error by optimizing over the probabilities $p_j$. In particular, by restricting our minimisation to only the leading order term of $\epsilon_1$ in $\epsilon$, 
we find that the leading order simulation error is
$r s^2 \sum_{j=1}^L h_j^2/ p_j $.

The way we find the optimal probability is by taking the derivative of the corresponding Lagrangian function, and thereby determine its turning points. If the primal and dual solutions are furthermore feasible and satisfy complementary slackness, then we know from the convexity of our problem that these primal and dual solutions are optimal for the primal and dual optimisation problems respectively. 

More formally, in convex optimisation theory, we know that a primal problem and its dual problem are both optimal if and only if (1) Slater's constraint qualification holds, and (2) the primal and dual variables satisfy the so-called Karush-Kuhn-Tucker (KKT) conditions \cite{nocedal2006numerical}. Of these two conditions (1) easily holds. The notion of the active set appears in one of the optimality conditions of (2), which is known as complementary slackness.

Now we only optimize over the $p_j$ for which $j$ belongs to the inactive set $\bar A$.
Hence we consider the optimisation problem 
\begin{mini} 
{p_j>0, j \in \bar A}{r s^2 \sum_{j \in \bar A} \frac{h_j^2}{p_j} }
{}{}
\addConstraint{\sum_{j \in \bar A} p_j }{=\bar \mu}
\addConstraint{p_j }{\le 1,}
\label{opt0}
\end{mini}
where $\bar \mu=\mu-|A|$.
Note that the objective function here is convex in $p_j$, 
and the constraint function is linear in $p_j$.
By treating $\bar \mu$ as a constant, 
we analytically derive the optimal value of this optimisation problem from the first order KKT conditions \cite{nocedal2006numerical}.
Since Slater's condition is satisfied, the KKT condition is necessary and sufficient for optimality. 
The KKT conditions require (1) the turning points of the Lagrangian to be zero, (2) primal feasibility, (3) feasibility of the Lagrange dual, (4) and complementary slackness. 
Complementary slackness requires the Lagrange multiplier of a constraint to be zero when that constraint is not tight. 

Denoting $u$ as the Lagrange multiplier for the equality constraint and $v_j$ as Lagrange multipliers for the inequality constraints, the Lagrangian of \eqref{opt0} is 
\begin{align}
{\rm L} 
&=  r s^2 \sum_{j \in \bar A} \frac{h_j^2}{p_j} + u \left(  \sum_{j \in \bar A} p_j   - \bar\mu  \right)  + \sum_{j \in \bar A} (p_j-1) v_j
\notag\\
&=   \sum_{j \in \bar A} \left( \frac{r s^2 h_j^2}{p_j} + (u+v_j) p_j - v_j  \right)  - u \bar \mu.
\end{align}
Note that
\begin{align}
\frac{\partial{\rm L} }{\partial p_j}
&= 
\frac{-r s^2 h_j^2}{p_j^2} + u +v_j,
\end{align}
and hence the turning point of the Lagrangian L occurs when 
\begin{align}
p_j = \frac{ \sqrt{r} s h_j }{\sqrt{u+v_j} }.
\end{align}
Note that we have $u \in \mathbb R$ and $v_j\ge 0$.
From complementary slackness, we know that if the optimal $p_j<1$, then we correspondingly have $v_j=0$. When $p_j=1$, the constraint corresponding to $v_j$ is active, and $v_j >0$.
Hence it follows that whenever $p_j < 1$, we have
\begin{align}
\sqrt{u} p_j = \sqrt{r} s h_j. 
\end{align}
Conversely, when $p_j =1$, we have
\begin{align}
u+v_j = {r s^2 h_j^2}.
\end{align} 
Note here that we have not verified that the problem is primal feasible, namely, that we need to check that $\sum_{j \in \bar A} p_j = \bar \mu$.
This can be satisfied whenever we have
\begin{align}
\bar \mu =s \sqrt{\frac r u} \sum_{j \in \bar A} h_j.
\end{align} 
Thus our ansatz for $p_j$ is 
\begin{align}
p_j &=    \frac{\bar\mu h_j}{  \sum_{j \in \bar A} h_j }  ,
\end{align}
with the regularity condition that this formula satisfies $p_j< 1$ for $j  \in \bar A$.

\end{document}